\documentclass[conference,10pt]{IEEEtran}
\pdfoutput=1
\usepackage{mathrsfs}
\IEEEoverridecommandlockouts
\makeatletter
\def\ps@headings{%
\def\@oddhead{\mbox{}\scriptsize\rightmark \hfil \thepage}%
\def\@evenhead{\scriptsizef\thepage \hfil \leftmark\mbox{}}%
\def\@oddfoot{}%
\def\@evenfoot{}}
\makeatother 
\usepackage[]{color}
\usepackage{amsmath,amssymb,stfloats,subfigure,multicol,graphicx}
\usepackage{epsfig,epsf,psfrag,amssymb,amsfonts,latexsym,slashbox,mathrsfs}

\newtheorem{theorem}{Theorem}
\newtheorem{lemma}{Lemma}

\usepackage[square, comma, compress, numbers]{natbib}

\usepackage{algorithm}
\usepackage{algorithmic}

\usepackage{blindtext}
\usepackage{etoolbox}
\makeatletter

\docsvlist{\@oddhead,\@evenhead,\ps@headings,\ps@IEEEtitlepagestyle,\ps@IEEEpeerreviewcoverpagestyle}
\makeatother

\begin{document}

\title{Maximum Likelihood Fusion of Stochastic Maps}
\author{Brandon Jones, \emph{Student Member, IEEE}, Mark Campbell, \emph{Member, IEEE} and Lang Tong, \emph{Fellow, IEEE}
\thanks{\scriptsize This work was supported in part by the Army Research Office under grant W911NF-10-1-0419.  Parts of this work were presented at the 49th Annual Allerton Conference on Communication, Control, and Computing, Monticello, Ill., Sept. 2011, and at the 2012 SPIE Defense and Security Symposium, Baltimore, MD., May 2012.}

\thanks{\scriptsize
B. Jones is with the School of Electrical and Computer Engineering, Cornell University, Ithaca, NY 14853, USA. Email:
{\tt bmj34@cornell.edu}}

\thanks{\scriptsize
M. Campbell is with the Sibley School of Mechanical and Aerospace Engineering, Cornell University, Ithaca, NY 14853, USA. Email:
{\tt mc288@cornell.edu}}

\thanks{\scriptsize
L. Tong is with the School of Electrical and Computer Engineering, Cornell University, Ithaca, NY 14853, USA. Email:
{\tt ltong@ece.cornell.edu}}

\thanks{\scriptsize Manuscript submitted 03/24/2013.}}

\markboth{submitted to IEEE Trans. Signal Processing}{}
\maketitle

\begin{abstract}  The fusion of independently obtained stochastic maps by collaborating mobile agents is considered. The proposed approach includes two parts: matching of stochastic maps and maximum likelihood alignment.  In particular, an affine invariant hypergraph is constructed for each stochastic map, and a bipartite matching via a linear program is used to establish landmark correspondence between stochastic maps.  A maximum likelihood alignment procedure is proposed to determine rotation and translation between common landmarks in order to construct a global map within a common frame of reference.  A main feature of the proposed approach is its scalability with respect to the number of landmarks: the matching step has polynomial complexity and the maximum likelihood alignment is obtained in closed form.  Experimental validation of the proposed fusion approach is performed using the Victoria Park benchmark dataset.
\end{abstract}

\section{Introduction} \label{Introduction}

\PARstart{T}{he} general nature of the problem under consideration is to construct a global map of landmarks from the individual efforts of collaborating agents that operate outside of a global frame of reference.  Each agent independently builds a vector of estimated landmark locations referred to as a \emph{stochastic map} \cite{Smith1988SMU,Smith1990ARV}.  Constructing a combined global map within a common reference frame from the individual maps of the agents is referred to as a problem of {\em fusion of stochastic maps}.   Intuitively, the problem has the interpretation of a mathematical jigsaw puzzle: the stochastic maps are the disoriented pieces and the sought after global map is the completed puzzle.

A benchmark scenario based on the Victoria Park dataset is illustrated in Fig. \ref{victoriaPlots}.  The satellite image shows the ground truth environment from which the stochastic maps are obtained.  Trees (landmarks) located in the park are mapped from the global reference frame of the environment to the individual local reference frames of the mobile agents.  Each agent thus has an independent, but partial model of the explored environment.  The shared objective of the agents is to build a more complete global model of the environment from the individually obtained local stochastic maps.

\begin{figure}[t!]
\centering
\includegraphics[width=1\linewidth]{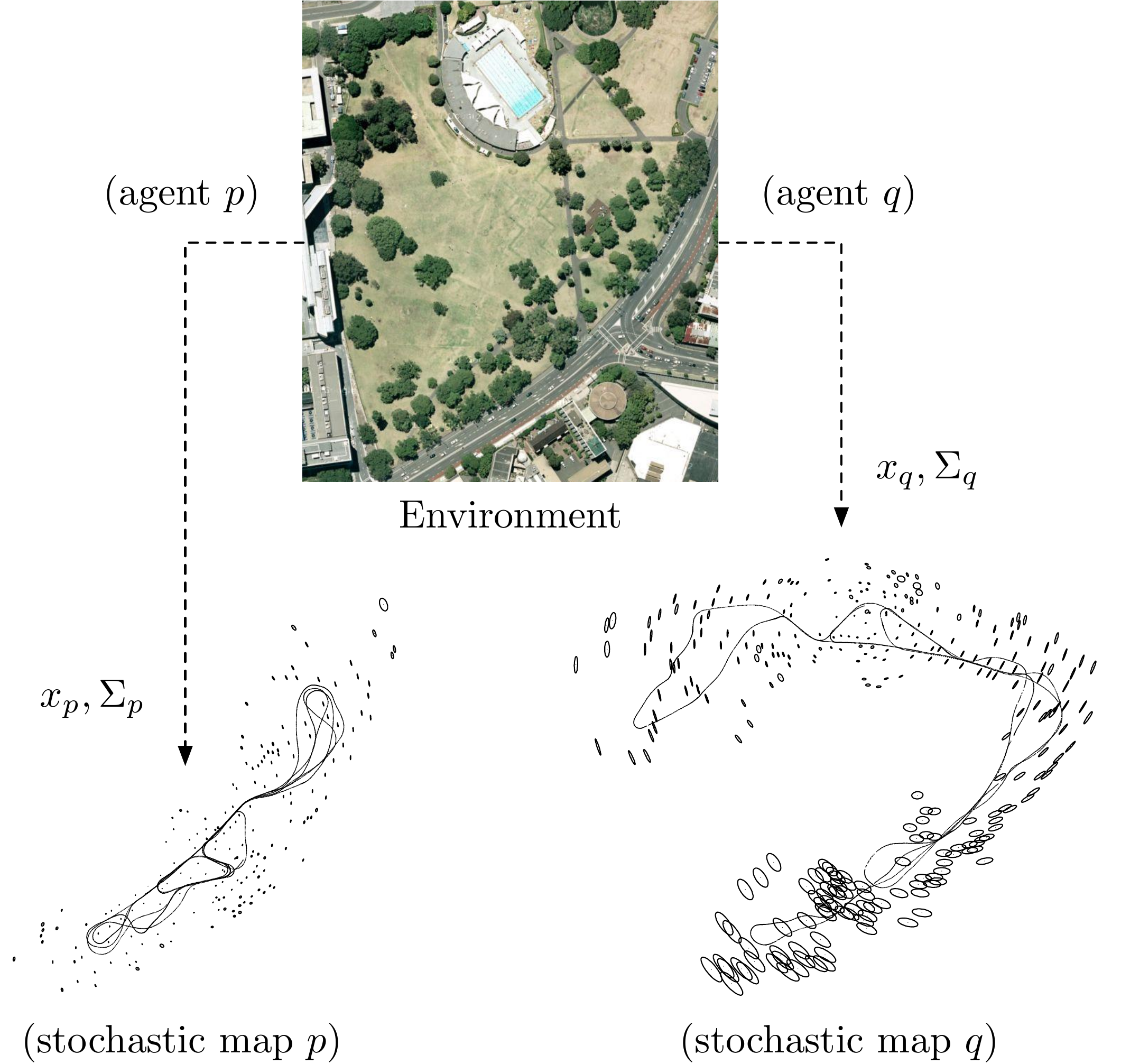}
\caption{Stochastic maps of the Victoria Park.  The stochastic maps of agent $p$ (bottom left) and agent $q$ (bottom right) contain landmark locations estimated by each agent.  Uncertainties in estimation are indicated by ellipses.  The path of exploration is shown by dotted lines.  Estimation of a common global map from the individual stochastic maps, each obtained in a separate coordinate system, requires inferring common landmarks in addition to determining a common frame of reference.}
\label{victoriaPlots}
\end{figure}

Stochastic maps are obtained by independent agents using various estimation techniques.  In robotics, the solution to the \emph{simultaneous localization and mapping (SLAM)} problem provides an agent with a stochastic map of the environment as a model of landmark locations (see \cite{Smith1988SMU, Smith1990ARV, Dissanayake2001TransRA, DurrantWhyte2006TransRA, Bailey2006TransRA} and the references therein).  The focus of this paper, however, is on the fusion -- rather than building -- of stochastic maps.  Our starting point is at the individual stochastic maps, which are made available to a fusion agent for the construction of a global map.

The fusion problem with multiple agents is challenging for several reasons, one being that the problem contains both discrete and continuous parts \cite{Thrun2003ISRR}.  In order to construct a global map, the fusion agent must first identify common landmarks residing in two separate maps.  Using the earlier jigsaw analogy, the solver has to first identify common edges in order to match the individual pieces.  Prior to exchanging stochastic maps, the agents are assumed to operate with no prior knowledge concerning the common landmarks (i.e., the common trees when considering the Victoria Park example) that are contained within the individual maps.  The problem of matching common landmarks is of a combinatorial nature in general, which eliminates exhaustive search as an option for large maps.

Even if common landmarks between two maps have been identified, the agents are faced with the \emph{alignment} problem of determining not only the best landmark estimates of common and uncommon landmarks contained by noisy maps obtained in separate coordinate systems, but also to determine the spatial parameters of rotation and translation.  Describing this again in terms of the earlier jigsaw analogy: not only are the pieces disoriented, but the edges are also imprecise (which makes it harder to see how the pieces fit together).  The alignment optimization is continuous in nature, but is also nonlinear and non-convex in general.

\subsection{Related work}

The matching and alignment problems considered in this paper have been studied in various forms.  Thrun and Liu \cite{Thrun2003ISRR} proposed an SR-tree (Sphere/Rectangle-tree) search \cite{Katayama1997SIS} in consideration of the matching problem.  Common landmark correspondences and rotation-translation parameters are found using an iterative hill climbing approach to match triplet combinations formed within a small radius of the landmarks in each map.  The radius forming the feature vectors of the SR-tree, however, would need to be adaptive in order to generalize to different environments.  Estimates of common landmarks are determined separately by a collapsing operation performed on matched landmarks in information form (see Grime and Durrant-Whyte \cite{Grime1994CEP}, as well as Sukkarieh \emph{et al.} \cite{Sukkarieh2003IJRR}, for further reading on fusion using information filtering).   Julier and Uhlmann \cite{Julier1997ACC} introduced the \emph{covariance intersection} algorithm as an approach to the data fusion problem.  Their algorithm uses a convex combination of state information to achieve data fusion, but has the limitation that the input data must be of the same dimension (which is often not the case of stochastic maps built within different regions of exploration).   Tard\'{o}s \emph{et al.} \cite{Tardos2002IJRR}, and later Castellanos \emph{et al.} \cite{Castellanos2007}, proposed \emph{map joining} as a technique to enable an individual mobile robot to construct a global stochastic map based on a sequence of local maps.  The approach is related to this paper by considering the sequence of local maps as being obtained from separate robots, but requires knowledge of a base reference to construct a global map.

Williams \emph{et al.} \cite{Williams2002ICRA} considered the fusion problem by providing parameter estimates of the relative rotation and translation between global and local maps.  The expressions are derived by observing the geometry of the landmarks within each map.  Our approach is distinct from \cite{Williams2002ICRA} in that the geometry of the landmarks is incorporated in a nonlinear least squares solution based on the maximum likelihood principle.  Several authors such as Zhou and Roumeliotis \cite{Zhou2006IROS}, Andersson and Nygards \cite{Andersson2008ICRA}, Benedettelli \emph{et al.} \cite{Benedettelli2010CDC}  and Aragues \emph{et al.} \cite{Aragues2011RAS} considered rendezvous approaches to the alignment problem.  Rendezvous approaches, however, are somewhat restrictive as the agents are required to be in close proximity.

The matching approach of this paper is motivated by the work of Groth \cite{Groth1986} and Ogawa \cite{Ogawa1986}.  Groth proposed one of the earliest matching algorithms in the context of astronomical point patterns, where a list of star measurements are matched against a known star catalog.  In the proposed approach, structured point triplets referred to simply as triangles are used to match the measurements against the catalog.  The Groth triangle convention is also incorporated in our approach, however the matching approach of Groth is not practical for large maps since all possible combinations of triangles are considered.  An alternative approach was proposed by Ogawa \cite{Ogawa1986}, which instead incorporated \emph{Delaunay triangulations} \cite{Delaunay1934} to address the star matching problem.  This paper therefore uses Delaunay triangulations with triangles that follow the Groth convention as a graphical model for matching stochastic maps.  Further insight into the structure of the model is found by arranging the triangles in order of increasing perimeter.

\subsection{Summary of results and organization}

A maximum likelihood framework is proposed for the construction of a global map from local stochastic maps. The proposed approach includes 1) a landmark matching approach referred to as \emph{generalized likelihood ratio matching (GLRM)} and 2) a least squares approach for jointly estimating rotation, translation and common landmark locations referred to as \emph{maximum likelihood alignment (MLA)}.  A Gaussian likelihood function is presented as the main proxy for deriving the procedures of each step.

Matching is a step that is performed in the absence of a global frame of reference, which requires a technique that is affine invariant.  To this end, the original stochastic maps are represented as directed hypergraphs constructed from Delaunay triangulations.  The hyperedges of each directed hypergraph are constructed from directed Delaunay triangles that follow the Groth convention, which leads to an affine invariant approach for determining common landmarks.  The proposed GLRM algorithm uses a generalized likelihood ratio as a matching metric in order to obtain globally optimal landmark correspondences from the solution of a bipartite matching problem.  The GLR metric is computed in closed form and the bipartite matching is solved in polynomial time as a solution to a  linear program.

Once common landmarks are identified, the solution to the alignment problem of determining rotation, translation and common landmark locations between two stochastic maps is computed from the determined common landmarks.  The main contribution is a closed-form solution to the alignment problem as nonlinear non-convex optimization, which makes optimal alignment trivial to obtain computationally.  

The remainder of the paper is organized as follows.  The problem formulation and models used throughout the paper are provided in Section \ref{formulation}.   While the alignment and matching steps share common likelihood functions, the maximum likelihood alignment problem is presented first in Section \ref{agreeModel} in order to introduce the proposed solution for closed form computations.  The problem of determining common landmarks is treated in Section \ref{sectionMatching}, where we present the GLRM approach.  Numerical examples and simulations are provided in Section \ref{numerical}.  The conclusion is given in Section \ref{conclusion} and is followed by an appendix of proofs.

\section{Model and problem formulation} \label{formulation}

\subsection{Ground truth model of landmarks}

A landmark is represented by a vector in ${\mathbb R}^2$ under a specific coordinate system.  Two collaborating agents $p$ and $q$ each estimate the locations of landmarks within a local frame of reference.  The coordinate systems of $p$ and $q$ are related by a rotation with parameter $\theta \in [-\pi,\pi]$ and a translation parameter $t \in {\mathbb R}^2$.  Specifically, if $\mu \in {\mathbb R}^2$ is the location of a landmark under coordinate system $p$, the landmark location under coordinate system $q$ is then

\[
\mu' = r(\theta)\mu + t,~~r(\theta)\triangleq \left[
\begin{array}{lr}
\cos \theta & -\sin \theta \\
\sin \theta & \cos \theta
\end{array}
\right].
\]
\vspace{-1pt}

\noindent
In the general case of $m$ landmark locations $\mu \in {\mathbb R}^{2m}$, again in coordinate system $p$, the representation in coordinate system $q$ is of the form 

\setlength{\arraycolsep}{0.3em}
\begin{equation*}
\mu'=R(\theta)\mu + Ft
\end{equation*}
\setlength{\arraycolsep}{5pt}
\vspace{-5pt}

\noindent
where $R(\theta) \triangleq I_m \otimes r(\theta)$ is the rotation matrix in block diagonal form and $I_m$ is an $m \times m$ identity matrix (the symbol $\otimes$ is the Kronecker product operator).  The matrix $F \triangleq e_m \otimes I_2$, with $e_m$ being an $m$-vector with all entries equal to $1$, applies the translation $t$ to each landmark in the map.

In this paper, without loss of generality, the ground truth of the combined map is defined in coordinate system $p$ by the vector $u=(\mu^T,v_p^T, v_q^T)^T$.  Common landmarks observed by both agents are contained by the vector $\mu \in {\mathbb R}^{2n}$, where $n$ is the number of common landmarks.  Landmarks observed by only agent $p$ are contained by the vector $v_p$ and landmarks observed by only agent $q$ are contained by the vector $v_q$.  The ground truth observed by agent $p$ is then $u_p = (\mu^T,v_p^T)^T$, which is defined in the coordinate system of agent $p$, and the ground truth of agent $q$ is $u_q = (\mu^T,v_q^T)^T$ in the global reference frame, which is observed in coordinate system $q$ as

\vspace{-8pt}
\setlength{\arraycolsep}{0.3em}
\begin{eqnarray} \label{deterministicQ}
 u_q' &=& R(\theta)u_q + Ft \nonumber\\
 &=& \left( \hspace{-2pt} \begin{array}{cc} R_1(\theta) & \vspace{2pt} \\ & \hspace{-5pt} R_0(\theta) \end{array} \hspace{-2pt} \right)
\left( \begin{array}{c} \mu \vspace{2pt} \\ v_q \end{array} \right) + \left( \begin{array}{c} F_1 \vspace{2pt} \\ F_0 \end{array} \right) t 
\end{eqnarray}
\setlength{\arraycolsep}{5pt}

\noindent
where the subscripts $1$ and $0$ indicate the partition the map into common and uncommon parts.

In addition to the inherent uncertainty of stochastic maps, the challenge of fusion as it relates to constructing a combined stochastic maps is that the parameters $\{t,\theta\}$ that relate the coordinate systems are unknown and the common landmarks observed by both agents are also unknown.  The ultimate goal of fusion is to estimate the combined map $u$ from the stochastic maps of the individual agents when the parameters $\{\mu,\theta,t\}$ are unknown.

\section{Maximum likelihood alignment} \label{agreeModel}

This section describes a maximum likelihood approach for constructing a global map of landmarks from stochastic maps obtained in separate coordinate systems.  We begin by describing a Gaussian model for the maps and propose a closed form solution to the maximum likelihood alignment problem of estimating the parameters $\{\mu,v_p,v_q,t,\theta\}$ under the assumption that the common landmarks between the maps are known.

\subsection{Matched Gaussian maps}

Let the random vectors $X_p$ and $X_q$ represent the noisy observations obtained by agents $p$ and $q$, respectively.  Prior to fusion, data is collected in the separate coordinate systems of the agents (i.e., $X_p$ and $X_q$ reside in coordinate systems $p$ and $q$, respectively).  The statistical model of {\em matched Gaussian maps} is given as

\vspace{-15pt}
\setlength{\arraycolsep}{0.3em}
\begin{eqnarray} \label{stochasticP}
X_p & = & u_p + W_p \vspace{5pt} \\
X_q & = & R(\theta) u_q + F t + W_q
\end{eqnarray}
\setlength{\arraycolsep}{5pt}
\vspace{-7pt}

\noindent
where $W_p \sim {\mathcal N}(0,\sigma_p^2I)$ and $W_q \sim {\mathcal N}(0,\sigma_q^2I)$ are independent zero-mean additive Gaussian noise vectors.  In order to separate the matching and alignment problems, an assumption is made that the common landmarks in both maps are known.  The process of obtaining such a matching, however, is nontrivial and combinatorial in general (see Section \ref{sectionMatching} for an affine invariant procedure for determining common landmarks).

\subsection{Likelihood decomposition and closed form solution}

Estimators of the parameters $\{\mu,v_p,v_q,t,\theta\}$ are derived by considering the likelihood function of the combined global map given by

\vspace{-5pt}
\setlength{\arraycolsep}{0.3em}
\begin{eqnarray} \label{likelihoodL}
L(\mu,v_p,v_q,t,\theta) 
& \triangleq & \eta \exp -\frac{1}{2} J(\mu,v_p,v_q,t,\theta)
\end{eqnarray}
\setlength{\arraycolsep}{5pt}
\vspace{-5pt}

\noindent
where $\eta$ is a normalizing constant and the function $J$, with unknown parameters as its arguments, is defined as

\setlength{\arraycolsep}{0.3em}
\begin{equation} \label{likelihoodJ}
J(\mu,v_p,v_q,t,\theta) \triangleq \frac{1}{\sigma_p^2} || x_p - u_p ||^2+ \frac{1}{\sigma_q^2}|| x_q - R(\theta) u_q - F t ||^2
\end{equation}
\setlength{\arraycolsep}{5pt}
\vspace{-5pt}

\noindent
with $x_p=(x_p^{1T},x_p^{0T})^T$ and $x_q=(x_q^{1T},x_q^{0T})^T$, corresponding to the structure of $u_p$ and $u_q$, respectively.  By partitioning the problem into common and uncommon parts, it immediately follows that (\ref{likelihoodJ}) decomposes as

\setlength{\arraycolsep}{0.3em}
\begin{equation}
J(\mu,v_p,v_q,t,\theta)=J_0(v_p,v_q,t,\theta)+J_1(\mu,t,\theta)
\end{equation}
\setlength{\arraycolsep}{5pt}
\vspace{-5pt}

\noindent
where $J_0$ is the squared error function of estimating the uncommon landmarks $v_p$ and $v_q$, including the transform parameters $\{t,\theta\}$, specified as

\vspace{-5pt}
\setlength{\arraycolsep}{0.3em}
\begin{eqnarray}
J_0( v_p,v_q,t,\theta) &\triangleq & \frac{1}{\sigma_p^2} ||x_p^0-v_p||^2 \nonumber\\
& &  + \frac{1}{\sigma_q^2} ||x_q^0-R_0(\theta)v_q-F_0t||^2
\end{eqnarray}
\setlength{\arraycolsep}{5pt}

\noindent
and $J_1$ is the squared error function of estimating the common landmarks contained by the vector $\mu$, also including $\{t,\theta\}$, specified as

\vspace{-5pt}
\setlength{\arraycolsep}{0.3em}
\begin{eqnarray}
J_1( \mu,t,\theta) &\triangleq & \frac{1}{\sigma_p^2} ||x_p^1-\mu||^2 \nonumber\\
& &  + \frac{1}{\sigma_q^2} ||x_q^1-R_1(\theta)\mu-F_1t||^2 \label{J1cost}.
\end{eqnarray}
\setlength{\arraycolsep}{5pt}

\noindent
This decomposition is exploited to minimize the combined error function $J$ by minimizing $J_0$ and $J_1$ separately, as stated by the following lemma.

\vspace{10pt} 
\begin{lemma}[Separable optimization] \label{lemmaDecompose}  Let $\{\mu^*, v_p^*, v_q^*, t^*,\theta^*\}$ be the global maximum of the likelihood function $L$, i.e.,

\vspace{-5pt}
\setlength{\arraycolsep}{0.3em}
\begin{equation}
J(\mu^*,v^*_p,v^*_q,t^*,\theta^*)=\min_{\mu,v_p,v_q,t,\theta} J(\mu,v_p,v_q,t,\theta).
\end{equation}
\setlength{\arraycolsep}{5pt}

\noindent
If the solution $\{{\hat \mu},{\hat t},{\hat \theta}\}$ is the global minimum of $J_1$ given by

\vspace{-8pt}
\setlength{\arraycolsep}{0.3em}
\begin{eqnarray}
({\hat \mu} , {\hat t} , {\hat \theta}) & = & \displaystyle \operatornamewithlimits{argmin}_{\mu , t , \theta} J_1( \mu , t , \theta ), \label{L1optimize}
\end{eqnarray}
\setlength{\arraycolsep}{5pt}

\noindent
then $\mu^* = {\hat \mu}$, $t^* = {\hat t}$, $\theta^* = {\hat \theta}$ and

\vspace{-3pt}
\setlength{\arraycolsep}{0.3em}
\begin{eqnarray}\label{vHat}
v_p^*=x_p^0,~~ v_q^*= R_0^T({\hat \theta})(x_q^0-F_0\hat{t})
\end{eqnarray}
\setlength{\arraycolsep}{5pt}
\vspace{-3pt}

\noindent
respectively.

\end{lemma}
\vspace{10pt} 
\begin{proof} With the decomposition $J=J_0+J_1$, the proof is immediate by noting that

\vspace{2pt}
\setlength{\arraycolsep}{0.3em}
\begin{equation}
J_0\left(x_p^0, R_0^T(\theta)(x_q^0-F_0t),t,\theta\right)=0
\end{equation}
\setlength{\arraycolsep}{5pt}
\vspace{-3pt}

\noindent
for any $\{\mu, t, \theta\}$.
\end{proof}

\vspace{10pt}
Lemma \ref{lemmaDecompose} shows that the maximum likelihood solution of the combined map specified by $u^* = (\mu^{*T}, v_p^{*T}, v_q^{*T})^T$ is obtained from the nonlinear least squares optimization of the non-convex function $J_1$.  A global minimum is obtained by deriving an equivalent expression of $J_1$ as a sinusoidal form parameterized by the unknown rotation parameter $\theta$ (see Appendix), which leads to a closed form solution as stated by the following theorem.

\pagebreak
\begin{theorem}[Closed form MLE] \label{mainTheorem}  The ML estimators of the parameters $\{\mu,t,\theta\}$ are given by the following expressions.

\begin{enumerate}

\item The MLE of the rotation parameter $\theta$ is

\setlength{\arraycolsep}{0.3em}
\begin{equation} \label{isoTheta}
\theta^* = \operatorname{sgn}(\beta) \left[ \cos^{-1} \left( \frac{\alpha}{\sqrt{\alpha^2+\beta^2}} \right) - \pi \right]
\end{equation}
\setlength{\arraycolsep}{5pt}

\noindent
where $\operatorname{sgn}(\cdot)$ is the signum function. The coefficients $\alpha$ and $\beta$ are given by

\vspace{-8pt}
\setlength{\arraycolsep}{0.3em}
\begin{eqnarray}
\alpha & = & - x_q^{1T} (I_n \otimes I_c) Q x_p^1 \\
\beta & = & - x_q^{1T} (I_n \otimes I_s) Q x_p^1
\end{eqnarray}
\setlength{\arraycolsep}{5pt}
\vspace{-8pt}
 
respectively, where $Q = I_{2n} - F_1(F_1^TF_1)^{-1}F_1^T$ with $n$ being the number of common landmarks (see Appendix for the constant matrices $I_c$ and $I_s$).

\item 
\noindent The MLE of the translation $t$ is

\begin{equation} \label{isoT}
t^*(\theta^*) = (F_1^TF_1)^{-1} F_1^T \left[ x_q^1 - R(\theta^*) x_p^1 \right]
\end{equation}
\vspace{-5pt}

denoted hereafter as $t^*$.

\item The MLE of the common landmarks $\mu$ is

\setlength{\arraycolsep}{0.3em}
\begin{equation} \label{isoFusion}
\mu^*(\theta^*) = \phi_p^* x_p^1 + \phi_q^* x_q^1
\end{equation}
\setlength{\arraycolsep}{5pt}
\vspace{-5pt}

denoted hereafter as $\mu^*$.  The matrix gains $\phi_p^*$ and $\phi_q^*$ are given by

\vspace{-8pt}
\setlength{\arraycolsep}{0.3em}
\begin{eqnarray}
\phi_p^* & = & I_{2n} - \frac{\sigma_p^2}{\sigma_p^2 + \sigma_q^2} Q \vspace{5pt} \\
\phi_q^* & = & \frac{\sigma_p^2}{\sigma_p^2 + \sigma_q^2} Q R^T(\theta^*)
\end{eqnarray}
\setlength{\arraycolsep}{5pt}

respectively.
\end{enumerate}
\end{theorem}

\vspace{5pt}
\begin{proof}
See Appendix.
\end{proof}
\vspace{5pt}

Theorem \ref{mainTheorem} specifies a closed form solution to the ML alignment problem using the realizations of matched Gaussian maps as data.  An important note, however, is that $Q = 0_{2 \times 2}$ when $n = 1$, meaning $n > 1$ common landmarks are required to compute the solution of Theorem \ref{mainTheorem} (a minimum of $n = 3$ common landmarks are recommended).

\section{Generalized Likelihood Ratio Matching} \label{sectionMatching}

In a general mapping scenario, the ground truth structure observed by the agents is unknown.  In particular, if the first two entries of $X_p = x_p$ correspond to the particular landmark, then the first two entries of $X_q = x_q$ correspond to a different landmark in general (and likewise with the remaining entries).  Common landmarks in this case are identified by applying a matching procedure to $x_p$ and $x_q$ with consideration that the stochastic maps are obtained in separate coordinate systems related by $\theta$ and $t$.  The matching procedure proposed in this section is based on the use of landmark triplets referred to as \emph{triangles}, which requires that the maps of each agent contain at least three landmarks.

\subsection{Directed hypergraph model} \label{sec:A}

Triangles are constructed from the maps of each agent by following a direction convention used in the star-pattern matching approach of Groth \cite{Groth1986}.  In particular, given three landmark locations $y_a, y_b, y_c \in {\mathbb R}^2$, the Groth representation of a directed triangle is $y = (y_a^T, y_b^T, y_c^T)^T$, which is a vector in ${\mathbb R}^6$ with entries that follow the inequality

\setlength{\arraycolsep}{0.3em}
\begin{equation}
|| y_a - y_b || < || y_b - y_c || < || y_c - y_a ||
\end{equation}
\setlength{\arraycolsep}{5pt}
\vspace{-5pt}

\noindent
under the assumption that no two triangle edges have the same length.  This convention, which is invariant to changes in rotation and translation is used to construct the \emph{directed hypergraphs} $G_p = (V_p,E_p)$ and $G_q = (V_q,E_q)$ from the Delaunay triangulations of maps $p$ and $q$, respectively.  The landmarks that form the vertices of each graph are contained by $V_p$ (agent $p$) and $V_q$ (agent $q$).  The resulting directed triangles constructed from the maps of $p$ and $q$ are contained by the hyperedges $E_p$ and $E_q$, respectively.

\subsection{Hypothesis testing and bipartite matching}
\vspace{5pt}

Determining common landmarks from the directed triangles of $G_p$ and $G_q$ is considered as a binary hypothesis testing problem.  Under hypothesis $H_0$, the agents observe the ground truth directed triangles $\nu_p,\nu_q \in {\mathbb R}^6$, which contain a maximum of two landmarks in common.  Under hypothesis $H_1$, the agents observe a common directed triangle $\delta \in {\mathbb R}^6$ within their respective coordinate systems.  In this way, $H_0$ is the hypothesis of uncommon triangles and $H_1$ is the hypothesis of common triangles.  The mathematical models of $H_0$ and $H_1$ are given by

\vspace{-5pt}
\setlength{\arraycolsep}{0.1em}
\begin{eqnarray}
H_0 &:&
\left[
\hspace{1pt}
\begin{array}{c}
Y_p \\
Y_q
\end{array}
\hspace{1pt}
\right]
\sim
{\mathcal N}
\left(
\left[
\begin{array}{c}
\hspace{1pt}
\nu_p \\
\nu_q
\end{array}
\hspace{1pt}
\right]
,
\left[
\begin{array}{cc}
\sigma_p^2 I &  \\
 & \hspace{2pt} \sigma_q^2 I
\end{array}
\right]
\right)\\
H_1 &:&
\left[
\hspace{1pt}
\begin{array}{c}
Y_p \\
Y_q
\end{array}
\hspace{1pt}
\right]
\sim
{\mathcal N}
\left(
\left[
\begin{array}{c}
\delta \\
R(\theta) \delta +Ft
\end{array}
\hspace{1pt}
\right]
,
\left[
\begin{array}{cc}
\sigma_p^2 I &  \\
 & \hspace{2pt} \sigma_q^2 I
\end{array}
\right]
\right)
\end{eqnarray}
\setlength{\arraycolsep}{5pt}

\noindent
respectively.  The appropriate matching hypothesis (i.e., $H_0$ or $H_1$) for the directed triangle data $Y_p = y_p$ and $Y_q = y_q$ is initially unknown.  Given the realizations $y_p$ and $y_q$ from the stochastic maps of $p$ and $q$, respectively, the matching hypothesis is determined using a \emph{generalized likelihood ratio test (GLRT)} of the form

\begin{figure}[t!]
\centering
\includegraphics[width=0.85\linewidth]{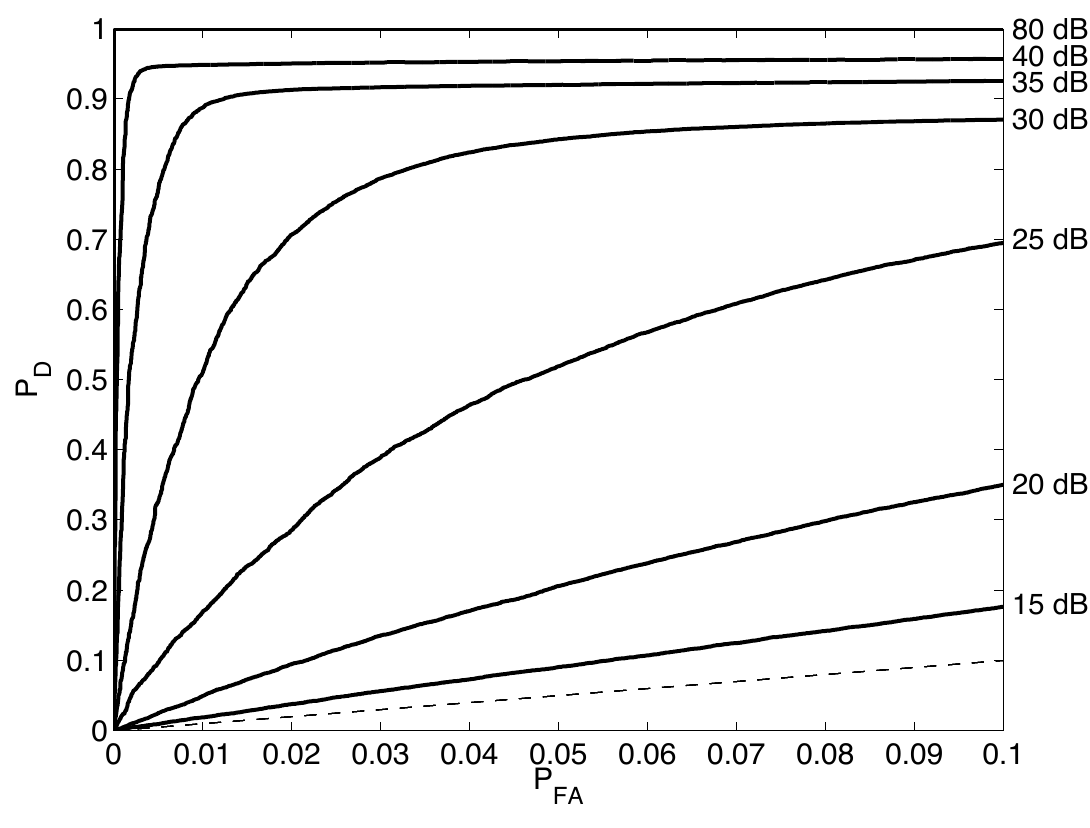}
\caption{Monte Carlo performance of hypergraph matching.  Receiver operating characteristic (ROC) curves, illustrated above, show the performance of the detecting triangle matches at various levels of SNR.  Each of the curves are plots of the probability of detecting a match ($P_\text{D}$) versus the probability of a false alarm ($P_\text{FA}$).  The dashed line in the lower region of the figure indicates the performance of a random guess.}
\label{singleTriROClabeled}
\end{figure}

\vspace{5pt}
\begin{equation}\label{eq:GLRT}
\Lambda(y_p,y_q) = \frac{ \displaystyle \max_{\delta,t,\theta} L_1(\delta,t,\theta)}
{ \displaystyle \max_{\nu_p,\nu_q,t,\theta} L_0(\nu_p,\nu_q,t,\theta)} \hspace{5pt} \operatornamewithlimits{\gtrless}_{H_0}^{H_1} \hspace{5pt} \tau
\end{equation}
\vspace{5pt}

\noindent
where $L_k$ are likelihood functions under $H_k$, with $k \in \{p,q\}$, and the threshold $\tau$ is selected to control the level of false alarm.  The likelihood statistic $\Lambda(y_p,y_q)$ is easily computed by applying Theorem \ref{mainTheorem}.  The performance of the approach in the presence of noise, as illustrated by the receiver operating characteristic (ROC) curves of Fig. \ref{singleTriROClabeled}, is of interest due to the uncertain nature of stochastic maps.  As illustrated in the figure, the performance of the approach degrades gracefully with increasing levels of noise (the signal-to-noise ratio, or SNR, is discussed in Section \ref{numerical}).

Applying the GLRT enables the determination of common triangles from data, but not in a one-to-one fashion as required to produce a consistent combined map.  If triangle $i \in \mathcal P$, with ${\mathcal P} = \{1,2,\hdots,|E_p|\}$, is denoted as $y_p^i$ and triangle $j \in \mathcal Q$, with ${\mathcal Q} = \{1,2,\hdots,|E_q|\}$, is denoted as $y_q^j$, then triangle matches are determined in a one-to-one fashion by formulating triangle matching as an assignment problem that seeks to

\begin{figure*}[t!]
\centering
\includegraphics[width=1\linewidth]{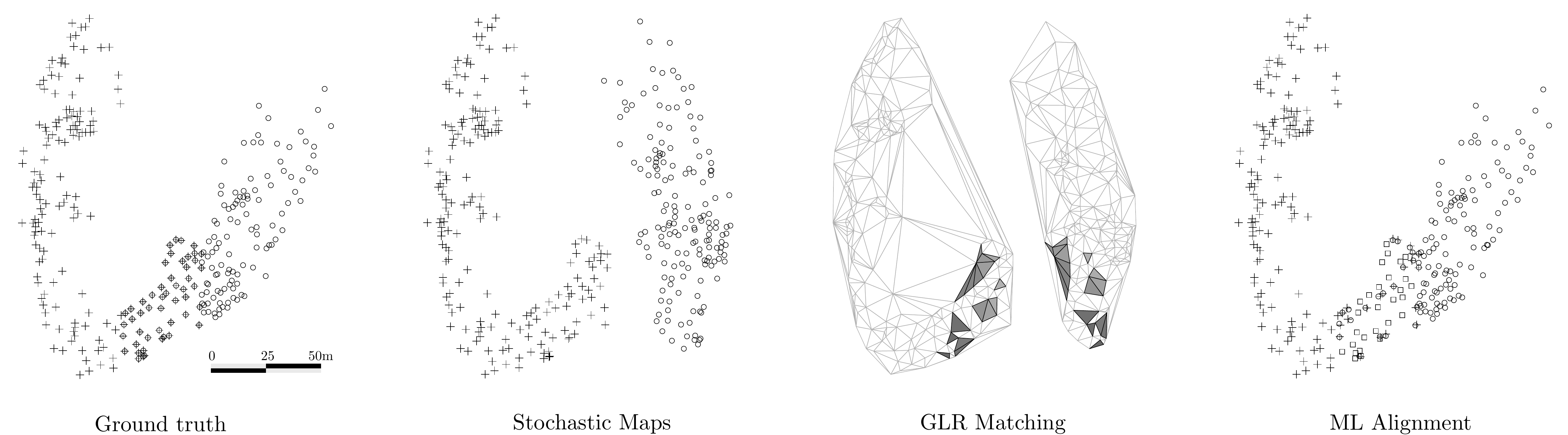}
\caption{Illustration of maximum likelihood fusion using the Victoria Park dataset.  (Ground truth) The true landmark locations observed by agent $p$ and agent $q$ are indicated by crosses (+) and circles ($\circ$), respectively.  The agents observe $50$ landmarks in common (contained by the vector $\mu$) with agent $p$ observing $179$ landmarks (contained by $u_p$) and agent $q$ observing $160$ landmarks (contained by $u_q$).  (Stochastic maps) The individual stochastic maps of the agents are generated using an additive Gaussian noise model, with the map of agent $q$ being transformed into a separate coordinate system by $\theta$ and $t$.  (GLR matching) Common landmarks are determined from a directed hypergraph representation of each stochastic map.  Using linear programming and outlier rejection, an inlier set of $16$ common directed triangles (shaded in gray) are used to estimate the parameters $\theta$ and $t$.  (ML alignment)  The closed form MLEs $\theta^*$ and $t^*$ (Theorem \ref{mainTheorem}) are used to compute the combined map $u^* = (\mu^{*T}, v_p^{*T}, v_q^{*T})^T$, where the common landmarks contained by $\mu^*$ are indicated by squares ({\tiny \text{$\square$}}).  Missed detections are easily detected in the common reference frame using nearest neighbor techniques.}
\label{palette4}
\end{figure*}
 
\setlength{\arraycolsep}{0.0em}
\begin{eqnarray} \label{integerProgram}
\hspace{-10pt} \text{maximize}  & \hspace{10pt} & \displaystyle \sum_{i=1}^m \sum_{j=1}^m f_{ij}(y_p^i,y_q^j) z_{ij} \\
\hspace{-10pt} \text{subject to} & \hspace{10pt} & \displaystyle \sum_{i=1}^m z_{ij} = 1, \hspace{10pt} j = 1,\hdots,m \\
\hspace{-10pt} 				& \hspace{10pt} & \displaystyle \sum_{j=1}^m z_{ij} = 1, \hspace{10pt} i = 1,\hdots,m \\
\hspace{-10pt} \text{and} 	 	& \hspace{10pt} & \displaystyle z_{ij} \in \{0,1\}, 
\end{eqnarray}
\setlength{\arraycolsep}{5pt}

\noindent
where $m = \max(|E_p|,|E_q|)$ and the function $f_{ij}$ used in the objective function is given by

\vspace{-5pt}
\setlength{\arraycolsep}{0.3em}
\begin{equation} \label{objVal}
f_{ij}(y_p^i,y_q^j) = \left\{
\begin{array}{ccl}
\Lambda(y_p^i,y_q^j) & , & i \in {\mathcal P} \text{ and } j \in {\mathcal Q} \vspace{5pt} \\
0 & , & \mbox{otherwise.}
\end{array}
\right.
\end{equation}
\setlength{\arraycolsep}{5pt}

\noindent
The structure of the assignment problem allows for the use of standard linear programming routines by relaxing the integer constraints to $z_{ij} \in [0,1]$.  The solution of the resulting linear program is indicated by the assignment set
 
\setlength{\arraycolsep}{0.3em}
\begin{equation}
{\mathcal A} \triangleq \left\{ (i \in {\mathcal P},j \in {\mathcal Q}) : z_{ij}^* = 1 \right\}
\end{equation}
\setlength{\arraycolsep}{5pt}
\vspace{-5pt}

\noindent
which includes one-to-one assignments of directed triangles to be identified as belonging to $H_0$ or $H_1$.  The GLRT (\ref{eq:GLRT}) provides a statistical approach for determining a matching hypothesis, but (as indicated in Fig. \ref{singleTriROClabeled}) the performance of the test degrades with increasing noise.  A robust detection scheme in the presence of uncertainty is to accept the assignments such that the MLEs $\{t_{ij}^*,\theta_{ij}^*\}$ of each triangle assignment form a consensus (under rigid body assumptions, the MLEs under $H_1$ form a cluster around the true values of $t$ and $\theta$).  Maximum likelihood estimation of the common parameters $\{\mu,t,\theta\}$ is then accomplished by applying Theorem \ref{mainTheorem} to the landmarks of the remaining accepted triangle assignments.

\section{Numerical examples and simulations} \label{numerical}

The fusion approach of this paper constructs a global map of landmarks in two main steps referred to as GLR matching and ML alignment.  This section provides an example of the ML fusion approach using the Victoria Park dataset and evaluates the performance of the matching and alignment steps in simulation.  An additional requirement of the matching step is outlier rejection, which is also discussed in this section.

\subsection{Illustration of maximum likelihood fusion}

Consider a Delaunay triangulation constructed from the ground truth landmarks observed by agent $p$ and agent $q$ with edge lengths $\{ \ell_i : i = 1,2,\hdots,m \}$.  By modeling the variance of the stochastic maps as $\sigma_p^2 = \sigma_q^2 = \sigma^2$, the signal-to-noise ratio (SNR) in decibels follows as

\begin{equation} \label{SNR}
\text{SNR}_\text{dB} = 10 \log \frac{\sigma_s^2}{\sigma_n^2}
\end{equation}

\vspace{5pt}
\noindent
with a signal variance of $\sigma_s^2 = \frac{1}{m} \sum_{i=1}^{m} \ell_i^2$ computed from the ground truth points and a noise variance of $\sigma_n^2 = 2 \sigma^2$ since the zero-mean Gaussian noise is additive to the landmarks rather than to the edge lengths directly.  The discussion of SNR in the remainder of the paper is in reference to (\ref{SNR}).

An illustration of the proposed ML fusion approach is shown in Fig. \ref{palette4}.  The ground truth landmarks shown in the figure are obtained by applying the the sparse local submap joining filter (SLSJF) proposed by Huang \emph{et al.} \cite{Huang2008TRO} to the Victoria Park dataset.  The ground truth is partitioned into two vectors $u_p$ and $u_q$ as models of the ground truth landmark locations observed by agent $p$ and agent $q$, respectively (see Section \ref{formulation}).  The stochastic maps of each agent are generated using the additive noise model discussed in Section \ref{agreeModel} at an SNR of $30\text{dB}$.  The rotation (in radians) and translation (in meters) applied to the stochastic map agent $q$ are $\theta = 0.7854$ and $t = (100,5)^T$, respectively.  As illustrated in the figure, $\theta$ and $t$ parameterize a spatial transform of the stochastic map of agent $q$ in reference to the ground truth coordinate frame (i.e., the coordinate system of agent $p$).

The directed hypergraph representation of each stochastic map is used by the GLR matching to determine the common directed triangles across the coordinate systems of the agents (using the Groth convention enables to the determination of common landmarks from common directed triangles).  The number of directed triangles in hypergraphs $p$ and $q$ are $|E_p| = 340$ and $|E_q| = 304$, respectively.  Due to the uncertainty of the stochastic maps, outlier rejection (discussed shortly) is required to determine an inlier set of matching triangles.  Theorem \ref{mainTheorem} is then applied to the inlier landmarks to compute the closed form MLEs $\theta^* = 0.7878$ and $t^* = (100.0860,4.9120)^T$.  Once the maps of each agent are represented within a common frame of reference, missed detections in the matching are easily found using common data association techniques such as nearest neighbor and maximum likelihood.  Maximum likelihood estimation of the combined map $u^* = (\mu^{*T}, v_p^{*T}, v_q^{*T})^T$ immediately follows from Lemma \ref{lemmaDecompose} and Theorem \ref{mainTheorem}.

\begin{figure}[t!]
\centering
\includegraphics[width=0.95\linewidth]{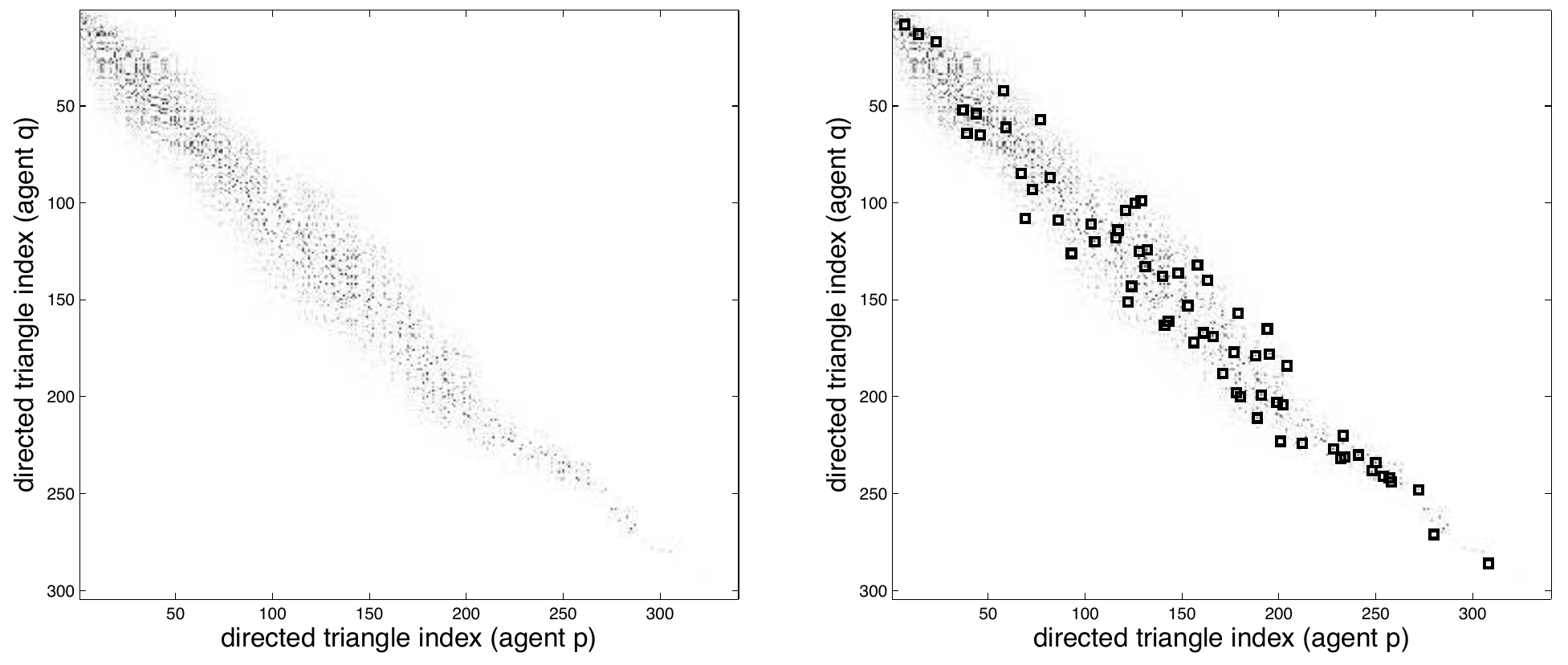}\\
\hspace{10pt} (a) \hspace{100pt} (b) \vspace{5pt} \\
\includegraphics[width=0.95\linewidth]{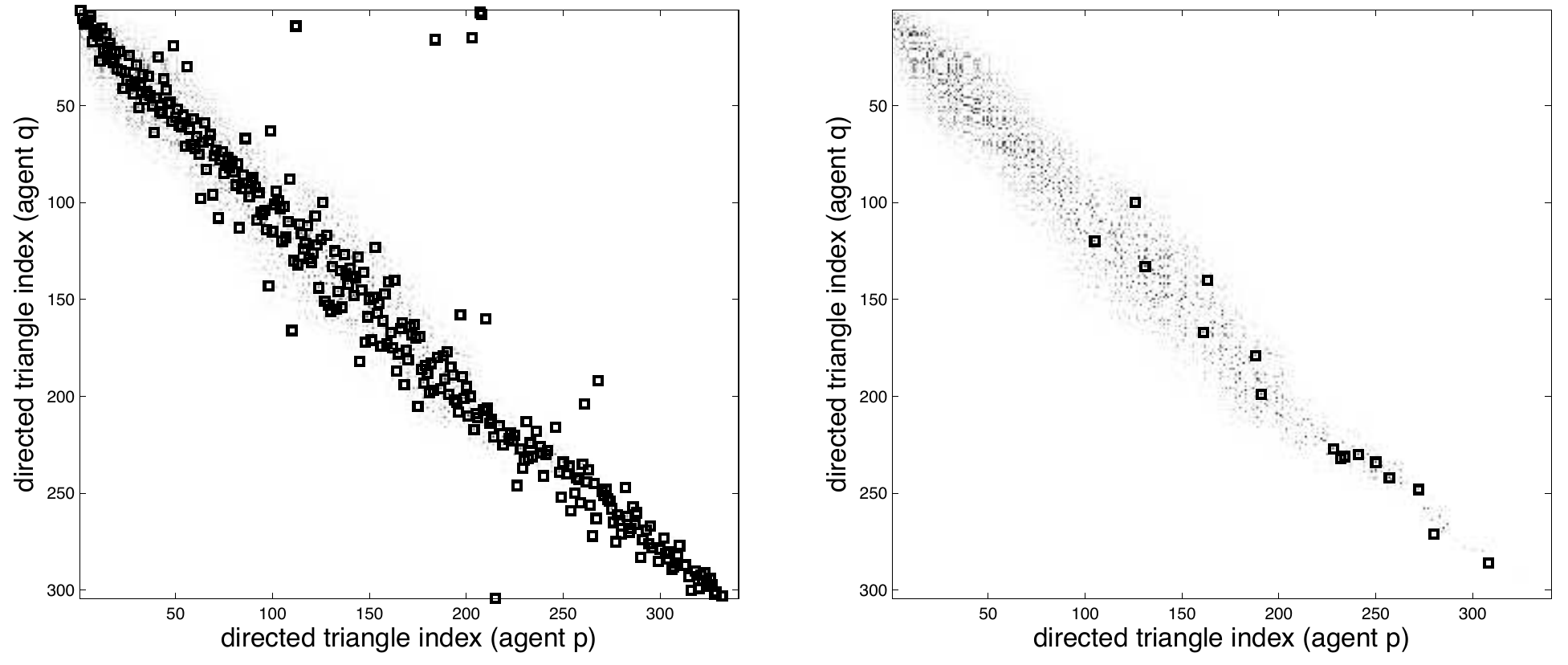}\\
\hspace{10pt} (c) \hspace{100pt} (d) \\
\caption{GLR matching statistics and triangle assignments (Victoria Park).  (a) The matrix of likelihood statistics is of the form $C = [\Lambda(y_p^i,y_q^j)]$, with the directed triangles of $p$ and $q$ indexed as $i \in {\mathcal P}$ and $j \in {\mathcal Q}$, respectively (the band diagonal structure of $C$ is due to arranging the triangles in order of increasing perimeter).  Darker entries indicate a higher likelihood of a directed triangle match.  Markers are used to indicated entry $C_{ij}$ corresponding to (b) the true directed triangle matches, (c) the triangle matches specified by the linear program and (d) the remaining 16 inlier matches that result from applying outlier rejection.}
\label{palette3}
\vspace{5pt}
\end{figure}

\begin{figure}[t!]
\centering
\includegraphics[width=0.8\linewidth]{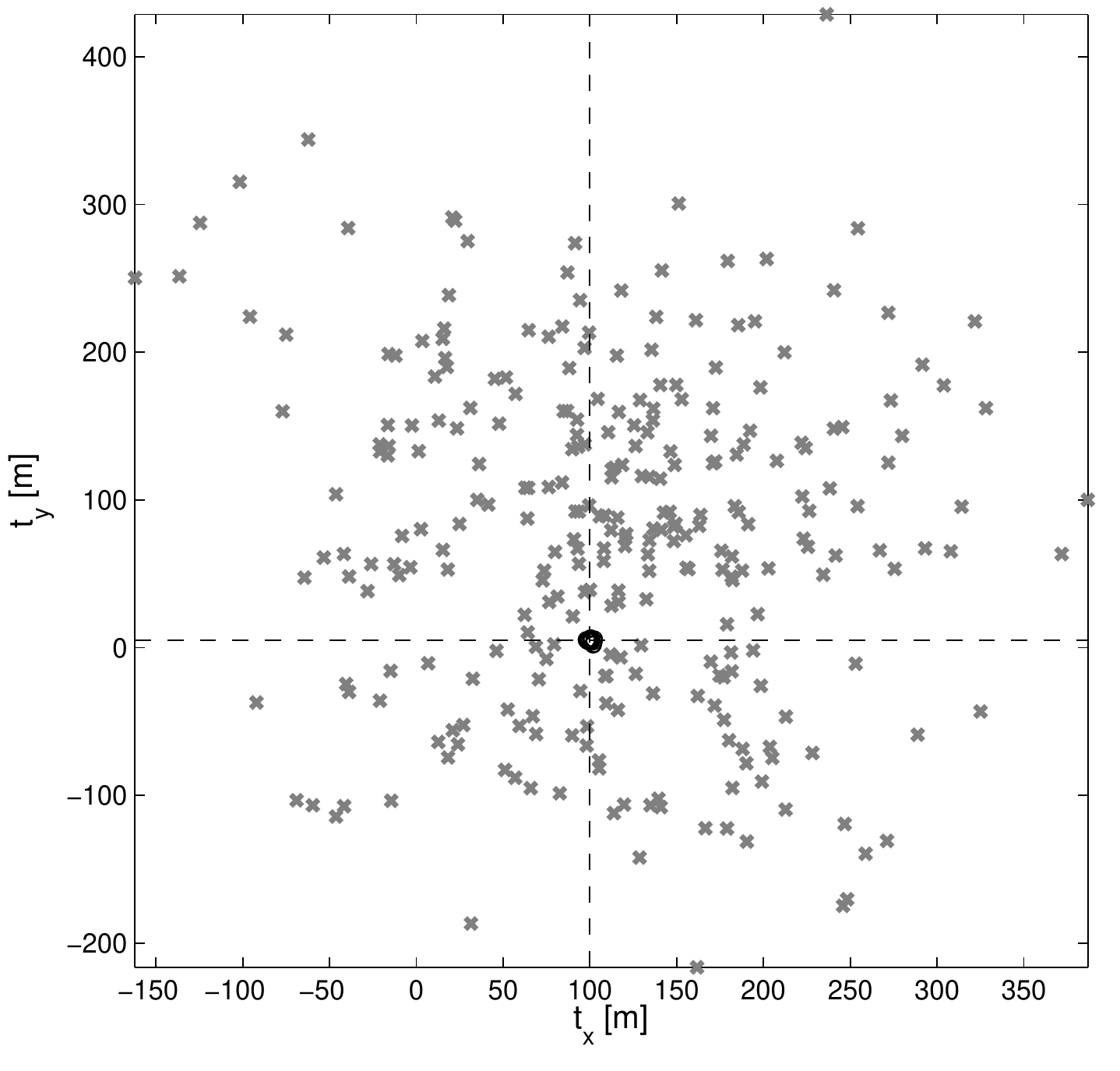}
\caption{Application of closed form MLE to outlier rejection (Victoria Park).  Each entry of the matrix $[\Lambda(y_p^i,y_q^j)]$ of likelihood statistics is associated with ML estimators of $\theta$ and $t$ of the form $\theta^*_{ij}$ and $t^*_{ij}$, respectively, with $i \in {\mathcal P}$ and $j \in {\mathcal Q}$.  The estimators of the parameter $t$ are shown in the 2D plot above (a similar plot is created in 3D by incorporating the ML estimators of the parameter $\theta$).  The true value of $t$ is indicated by the intersection of the dashed lines.  Inlier matches (black circles) are indicated by the MLEs that form a cluster around the true value of $t$, with outliers (gray crosses) indicated by the entries that fall outside the cluster.  The inlier set is indicated by $16$ closed form MLEs that cluster around the true values of $\theta$ and $t$.}
\label{outlierPlot}
\end{figure}

\begin{figure}[b!]
\centering
\includegraphics[width=0.75\linewidth]{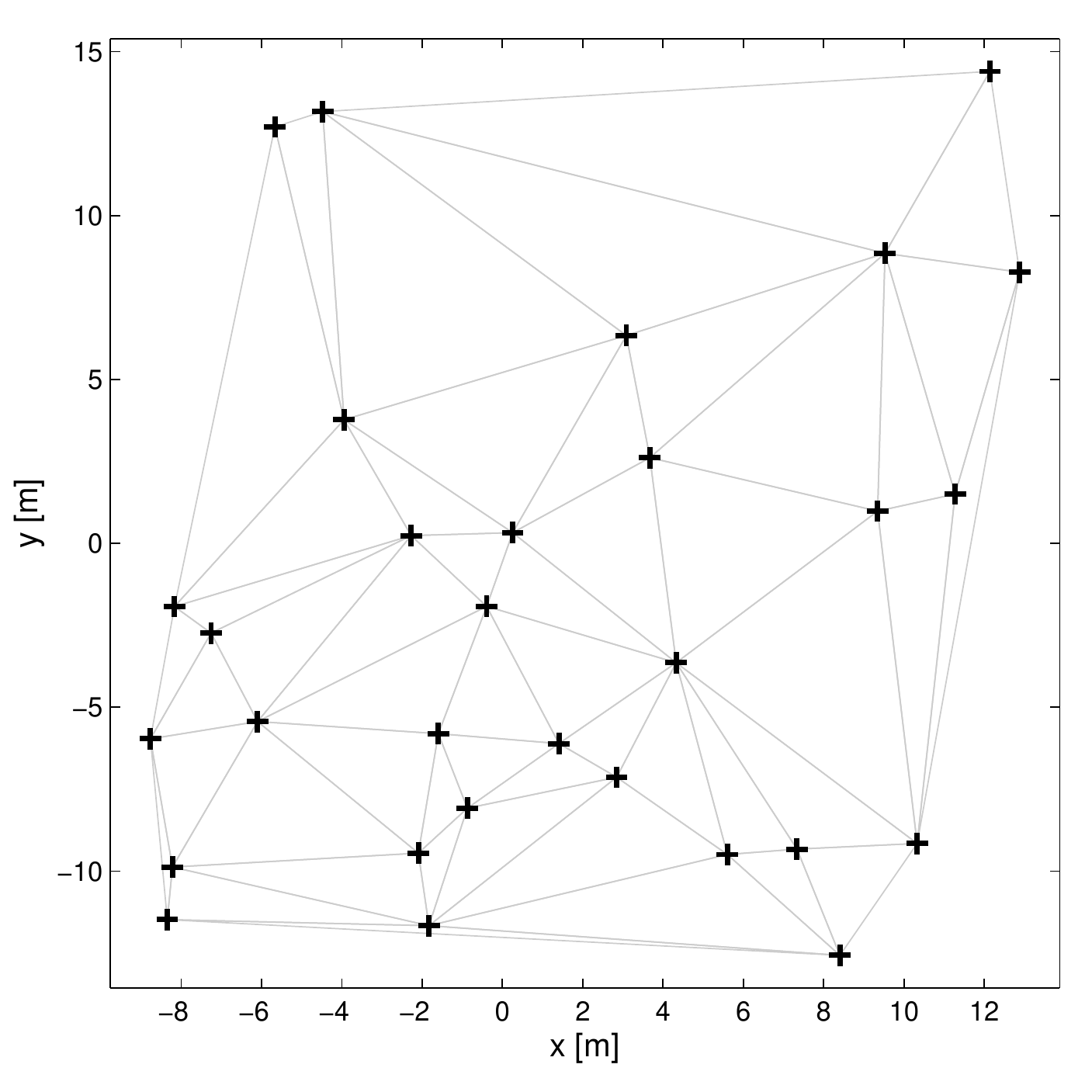} \\
\hspace{12pt} (a) \\
\hspace{5pt} \includegraphics[width=0.725\linewidth]{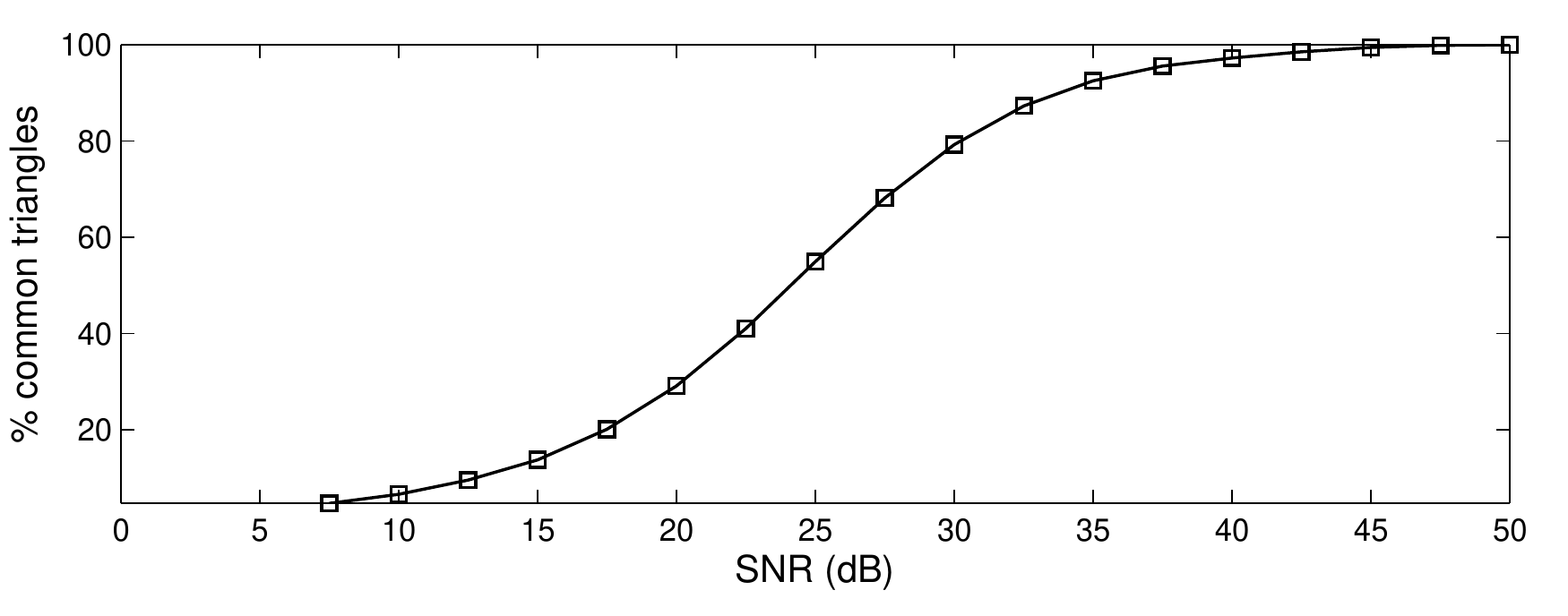} \\
\hspace{10pt} (b)
\vspace{-5pt}
\caption{Ground truth simulation and common directed triangles. (a) Landmarks contained by the ground truth vector $u$ with $x$- and $y$- coordinates being samples drawn from a uniform distribution. (b) The stochastic maps of $p$ and $q$ are generated from the ground truth with complete overlap, however the percentage of common triangles declines with increasing noise (as shown in the Monte Carlo simulation above).}
\label{MSEmonte}
\end{figure}

Inliers and outliers of the assignment set ${\mathcal A}$ are identified as follows.  Consider the matrix of likelihood statistics shown in Fig. \ref{palette3}, which is used by the linear program to construct the assignment set (the true directed triangle matches, as well as the matches contained by ${\mathcal A}$, are indicated by markers).  As illustrated in the figure, the linear programming formulation computes $304$ triangle assignments, however there are only $66$ common triangles in truth.  Each entry of ${\mathcal A}$ is associated with closed form MLEs of the parameters $\theta$ and $t$ of the form $\theta^*_{ij}$ and $t^*_{ij}$, respectively, as illustrated in Fig. \ref{outlierPlot} (only the ML estimators of $t$ are shown to simplify the illustration).  A simple heuristic for determining the inlier set is to recursively reject the individual MLEs with the largest sample MSE relative to the sample mean.  The rejection is repeated until the sample variance of the remaining estimators falls below a specified threshold, as indicated by the $16$ MLEs accepted as inliers in Fig. \ref{outlierPlot}.  This heuristic is preferred over the more standard RANSAC algorithm \cite{Fischler1981ACM} due to the conservative nature of the heuristic in choosing the inlier set (the inlier set is known \emph{a priori} to form a cluster under rigid body assumptions) in addition to the guarantee that the number of iterations of the heuristic is no greater than the cardinality of ${\mathcal A}$.  The corresponding accepted matches are shown at the bottom right of Fig. \ref{palette3}, which subsequently leads to the GLR matching diagram of Fig. \ref{palette4}.

\begin{figure}[b!]
\centering
\includegraphics[width=0.9\linewidth]{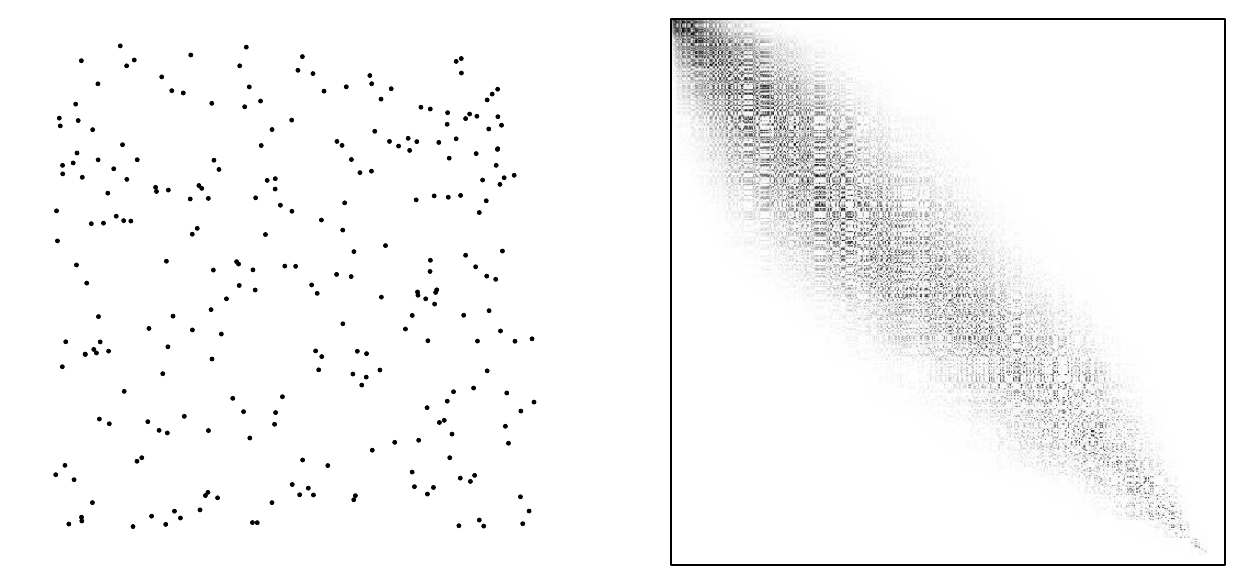} \\
 \vspace{1pt} (a) \vspace{3pt} \\
\includegraphics[width=0.9\linewidth]{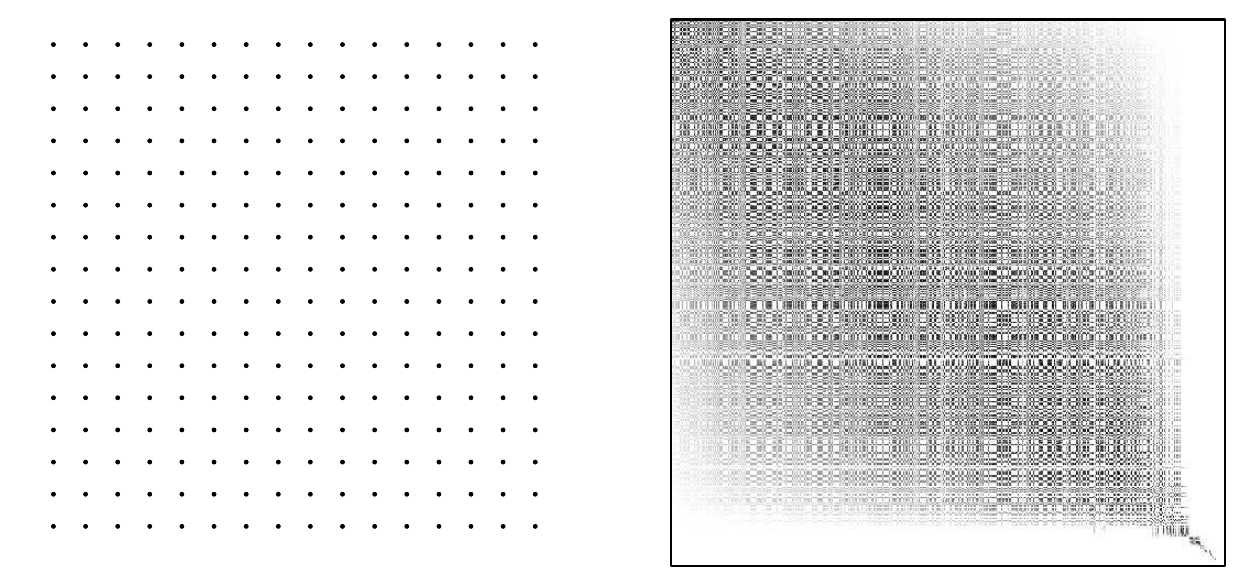} \\
 \vspace{1pt} (b)
\caption{Ground truth environments and hypergraph matching statistics.  (a) A ground truth environment, shown a the left, is generated using $256$ points drawn from a uniform distribution.  At the right are GLR matching statistics of stochastic maps generated from the ground truth at $30$dB. (b) Another ground truth environment of $256$ points is generated as a deterministic grid.  Unlike the uniform case, the matching statistics in this case indicate a large number of equally likely directed triangle matches.}
\label{statisticsA}
\end{figure}

\subsection{Performance with simulated ground truth}

The performance of GLR matching is considered, followed by the performance of ML alignment.  A simple simulation of ground truth is shown in Fig. \ref{MSEmonte}.  The ground truth vector $u$ contains $30$ landmark locations drawn as samples from the uniform distribution.  The stochastic maps of $p$ and $q$ are generated from the ground truth with complete overlap.  The GLR matching approach uses Delaunay triangulations to compute a generalized likelihood ratio as a metric for matching directed triangles.  Fig. \ref{singleTriROClabeled} shows the decline in performance of the likelihood statistic to determine common triangles with decreasing SNR.  The performance of the underlying Delaunay triangulations is considered in Fig. \ref{MSEmonte}.  The simulation shows that a gradual decline in the percentage of common triangles can be expected with decreasing SNR even in a scenario where the stochastic maps have a complete overlap in ground truth.

Perhaps a more subtle consideration in evaluating the GLR matching is the distribution of the ground truth landmarks.  Two simple examples are shown in Fig. \ref{statisticsA}.  In the first example, the ground truth is generated as samples from a uniform distribution.  In this case, the GLR statistics computed from stochastic maps of the ground truth exhibits a banded structure similar to the Victoria Park example.  In such an environment, the GLR matching approach is capable of determining common directed triangles between the stochastic maps.  In the second example, however, the ground truth is generated as a deterministic grid.  In this case, the banded structure of the likelihood statistics is lost due to a significantly large number of equally likelihood triangles.  The performance of the GLR matching in this case is severely impacted not necessarily due to sensor and process noise, but rather to the type of environment explored by the mobile agents.  For this reason, the deterministic grid serves as a counter example of the GLR matching approach.

\begin{figure}[t!]
\centering
\includegraphics[width=1\linewidth]{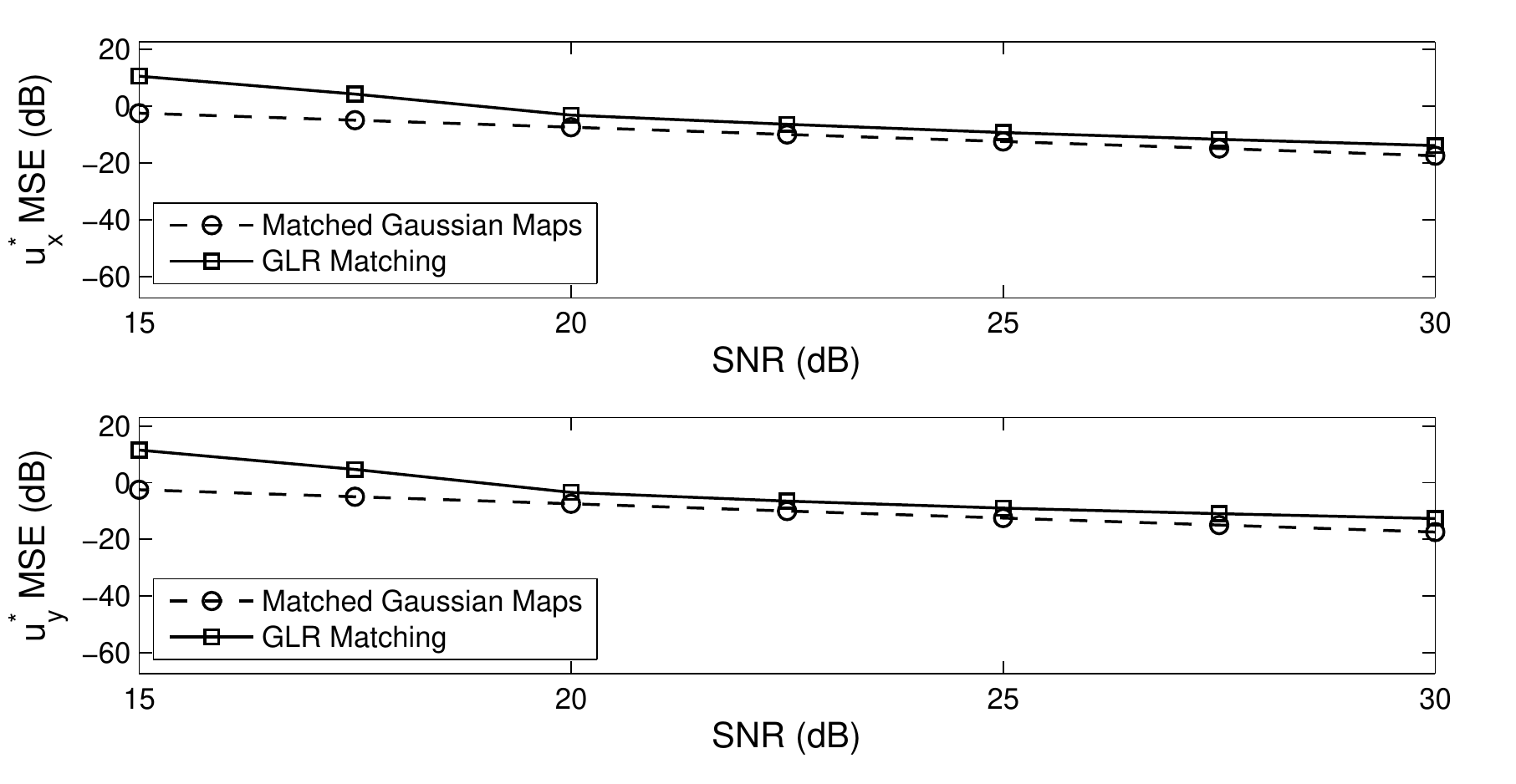} \\
(a) \\
\includegraphics[width=1\linewidth]{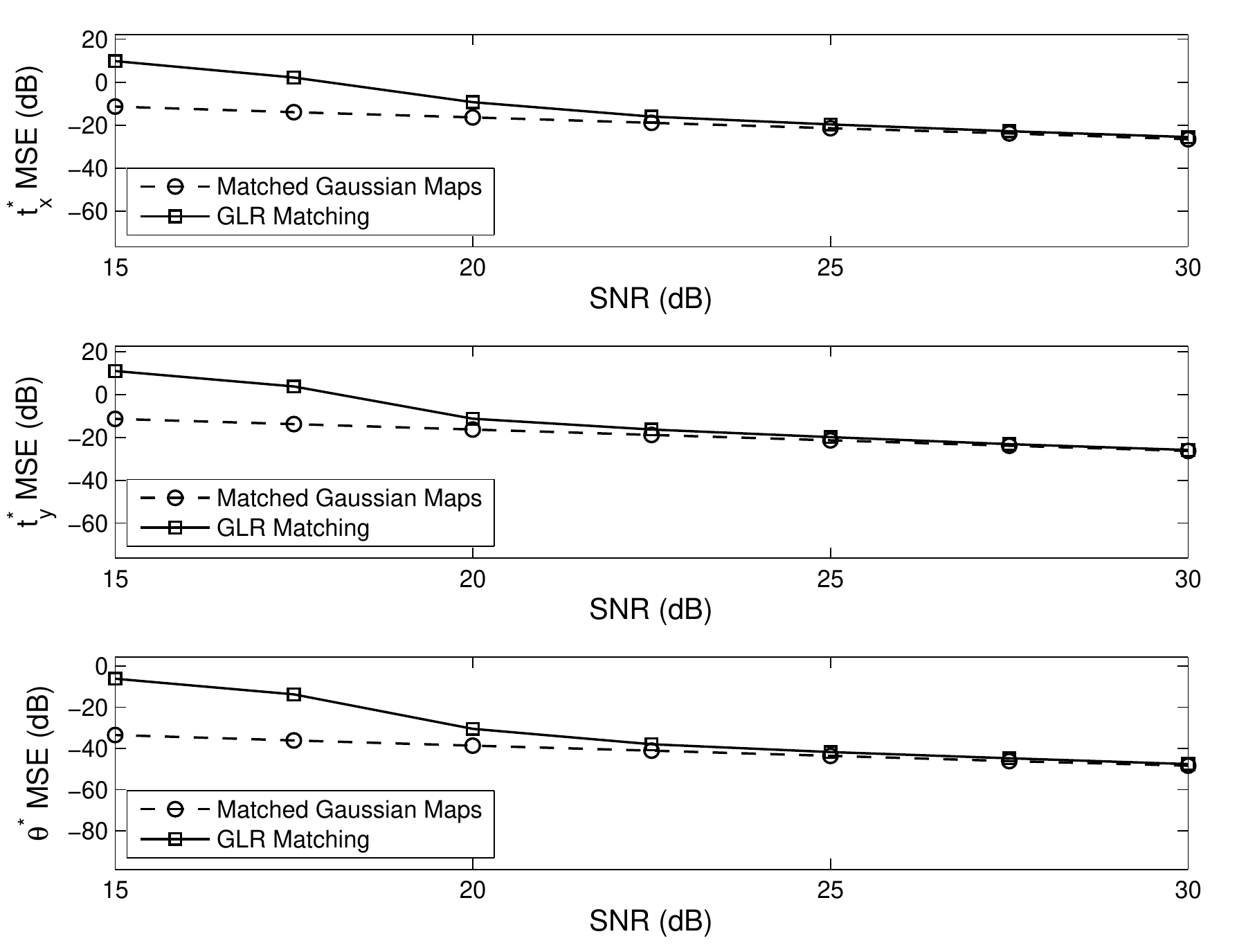} \\
(b)
\caption{Monte Carlo performance of closed form MLE.  The MSE (in dB) of estimating (a) the combined map $u$ and (b) the transform parameters $t$ and $\theta$ are shown in the figures above.  The performance of the ML estimators under the model of matched Gaussian maps, which the common landmarks are known, is shown by the dashed lines.  In the general case that the common landmarks are unknown, the estimators are computed from an inlier set of matched directed triangles as determined by the GLR matching.}
\label{MSEplots}
\end{figure}

The performance of the ML alignment approach is shown in Fig. \ref{MSEplots}.  The Monte Carlo simulation of the figure is based on stochastic maps generated from the ground truth landmarks shown in Fig. \ref{MSEmonte}.  Given the common landmarks in each map, the MSE performance of the estimators exhibit are linear trend in performance degradation with decreasing SNR.  A similar trend is observed in the case of unknown landmarks when the variance of the noise is relatively low.  At lower SNR, however, a greater decrease in MSE performance is observed relative to the known landmark case due not only to a decreased performance of the GLR matching statistic (Fig. \ref{singleTriROClabeled}), but also to a decrease in common triangles due to noise (Fig. \ref{MSEmonte}).  An additional decrease in performance is due to missed detections in the GLR matching (as seen in the Victoria Park example at the right of Fig. \ref{palette4}).  As SNR increases, however, the plots exhibit a convergence in MSE performance.

\section{Conclusion} \label{conclusion}

This paper considered the problem of constructing a global map of landmarks from the stochastic maps of collaborating agents -- fusion of stochastic maps. The problem can be formulated as a mixed integer-parameter estimation problem from which landmarks common to each agent is aligned under a global coordinate system.  Under this framework, the optimal fusion of stochastic maps can be accomplished using the maximum likelihood principle. Unfortunately, however, the complexity of the true ML solution is prohibitive, which leads to a partitioning of the problem into two steps: (i) matching landmarks via bipartite hypergraph matching using generalized likelihood ratio as a quality measure and (ii) maximum likelihood alignment to obtain the estimated of the combined map under a common coordinate frame.

The main advantage of the proposed approach, in spite of its suboptimality due to the separate treatment of the matching (which includes outlier rejection) and alignment problems, is the comprehensive nature of the procedure: a global map is found in spite of the individual stochastic maps being obtained in separate coordinate systems without prior knowledge of common landmarks. In simulations, the performance of the proposed approach is reasonable at high SNR but deteriorates with increasing noise, which is largely due the matching step.  One way to improve the performance may be to impose neighborhood constraints on triangles in addition to better heuristics for removing outliers.  Further computation gains may also be found by exploiting the banded structure of the likelihood matching statistics.

\section*{Acknowledgements}

The authors would like to thank Jose Guivant, Juan Nieto and Eduardo Nebot for providing the Victoria Park benchmark dataset.  The satellite image of Fig. \ref{victoriaPlots} was captured using Google Maps with a latitude-longitude coordinate of $(-33.886577,151.192061)$.

\appendix

The constant $2 \times 2$ matrices $I_c$ and $I_s$ used in Theorem \ref{mainTheorem} are defined as

\vspace{5pt}
$$I_c \triangleq
\left[
\begin{array}{lr}
1 & 0 \\
0 & 1
\end{array}
\right] \hspace{10pt} \text{ and } \hspace{10pt} I_s \triangleq
\left[
\begin{array}{lr}
0 & -1 \\
1 & 0
\end{array}
\right]$$
\vspace{2pt}

\noindent
respectively, with nonzero entries corresponding to the cosine and sine functions of the rotation matrix $r(\theta)$.  In addition, the following lemma is used in the proof of Theorem \ref{mainTheorem}.
\vspace{10pt}

\begin{lemma} \label{Qlemma}
The matrix $Q$ is an idempotent and symmetric matrix that commutes with a block diagonal matrix of the form $A = I_n \oplus B$, with $B \in {\mathbb R}^{2\times 2}$.
\end{lemma}

\vspace{10pt}
\begin{proof}[Proof of Lemma \ref{Qlemma}] 
The idempotence and symmetry properties are immediate from the structure of the matrix $Q$.  The product of the matrix $FF^T$ and the block diagonal matrix $A = I_{n} \otimes B$ is given by

\setlength{\arraycolsep}{0.3em}
\begin{equation*}
\begin{array}{rcl}
FF^T A & = & [(e_{n} e_{n}^T) I_{n}] \otimes [I_2 B] \vspace{5pt} \\
& = & [I_{n} (e_{n} e_{n}^T)] \otimes [B I_2] \vspace{5pt} \\
& = & [I_{n} \otimes B] [(e_{n} e_{n}^T) \otimes I_2] \vspace{5pt} \\
& = & A FF^T.
\end{array}
\end{equation*}
\setlength{\arraycolsep}{5pt}

\noindent
Since $FF^T A = A FF^T$ and $F^TF = \frac{1}{n} I_2$, it follows that

\setlength{\arraycolsep}{0.3em}
\begin{equation*}
\begin{array}{rcl}
QA & = & [I_{2n} - F(F^TF)^{-1}F^T] A \vspace{5pt} \\
& = & A - \frac{1}{n}FF^T A \vspace{5pt} \\
& = & A - \frac{1}{n}A FF^T \vspace{5pt} \\
& = & A [I_{2n} - F(F^TF)^{-1}F^T] \vspace{5pt} \\
& = & A Q
\end{array}
\end{equation*}
\setlength{\arraycolsep}{5pt}

\noindent
which proves that $Q$ and $A$ commute. \end{proof}

\vspace{10pt} 
\begin{proof}[Proof of Theorem \ref{mainTheorem}] Minimizing (\ref{J1cost}) with respect to (w.r.t.) $\mu$ leads to

\vspace{-5pt}
\setlength{\arraycolsep}{0.3em}
\begin{equation} \label{scalarMufixedThetaBar}
{\bar \mu}(\theta,t) = \frac{\sigma_p^2 \sigma_q^2}{\sigma_p^2 + \sigma_q^2} \left[ \frac{1}{\sigma_p^2} x_p + \frac{1}{\sigma_q^2} R^T(\theta)(x_q - Ft)\right]
\end{equation}
\setlength{\arraycolsep}{5pt}

\noindent
and minimizing (\ref{J1cost}) w.r.t. $t$ leads to

\vspace{-5pt}
\setlength{\arraycolsep}{0.3em}
\begin{eqnarray} \label{scalarTBar}
{\bar t}(\theta,\mu) = (F^TF)^{-1} F^T\left( x_q - R(\theta) \mu \right).
\end{eqnarray}
\setlength{\arraycolsep}{5pt}
\vspace{-5pt}

\noindent
Using the evaluation $\mu = {\bar \mu}(\theta,t)$ in (\ref{scalarTBar}) results in the MLE of $t$ as a function of $\theta$ given by

\setlength{\arraycolsep}{0.0em}
\begin{equation} \label{fixedOrientationThetaT}
t^*(\theta) = (F^TF)^{-1} F^T\left( x_q - R(\theta) x_p \right).
\end{equation}
\setlength{\arraycolsep}{5pt}
\vspace{-5pt}

\noindent
Applying the evaluation $t = t^*(\theta)$ in (\ref{scalarMufixedThetaBar}) leads to the MLE of $\mu$ as a function of $\theta$ given by

\vspace{-5pt}
\setlength{\arraycolsep}{0.3em}
\begin{eqnarray} \label{scalarMufixedTheta}
\mu^*(\theta) = \phi_p(\theta) x_p + \phi_q(\theta) x_q.
\end{eqnarray}
\setlength{\arraycolsep}{5pt}
\vspace{-5pt}

\noindent
Using the symmetry and idempotence properties of the matrix $Q$ (Lemma \ref{Qlemma}), it follows from the expression (\ref{scalarMufixedTheta}) that

\setlength{\arraycolsep}{0.3em}
\begin{eqnarray*}
\begin{array}{rcl}
|| x_p-\mu^*(\theta) ||^2 & = & \kappa_p || x_q - R(\theta) x_p ||_{Q^T Q}^2 \vspace{5pt} \\
& = & \kappa_p || x_q - R(\theta) x_p ||_Q^2
\end{array}
\end{eqnarray*}
\setlength{\arraycolsep}{5pt}

\noindent
where $\kappa_p = \left( \frac{\sigma_p^2}{\sigma_p^2 + \sigma_q^2} \right)^2$.  Similarly, it follows from (\ref{fixedOrientationThetaT}) and (\ref{scalarMufixedTheta}) that

\vspace{-5pt}
\setlength{\arraycolsep}{0.3em}
\begin{eqnarray*}
|| x_q-R(\theta)\mu^*(\theta)-Ft^*(\theta) ||^2 & = & \kappa_q || x_q - R(\theta) x_p ||_{Q^T Q}^2 \vspace{5pt} \\
& = & \kappa_q || x_q - R(\theta) x_p ||_Q^2
\end{eqnarray*}
\setlength{\arraycolsep}{5pt}
\vspace{-5pt}

\noindent
where $\kappa_q = \left( \frac{\sigma_q^2}{\sigma_p^2 + \sigma_q^2} \right)^2$.  Using these simplifications to define

\setlength{\arraycolsep}{0.3em}
\begin{eqnarray} \label{auxCost1}
\begin{array}{rcl}
J_1^*(\theta) & \triangleq & \frac{1}{2\kappa} J_1( \mu^*(\theta),t^*(\theta),\theta ) \vspace{5pt} \\
& = & \frac{1}{2} || x_q - R(\theta) x_p ||_Q^2
\end{array}
\end{eqnarray}
\setlength{\arraycolsep}{5pt}

\noindent
where $\kappa = \frac{1}{\sigma_p^2} \kappa_p + \frac{1}{\sigma_q^2} \kappa_q$ and expanding the norm in the right hand side (RHS) of (\ref{auxCost1}) as

\vspace{-5pt}
\setlength{\arraycolsep}{0.3em}
\begin{equation} \label{expanded}
\begin{array}{rcl}
|| x_q - R(\theta) x_p ||_Q^2 & = & x_p^TR^T(\theta) QR(\theta)x_p + x_q^TQx_q \vspace{5pt} \\
& & - 2 x_q^T QR(\theta)x_p, \vspace{5pt} \\
\end{array}
\end{equation}
\setlength{\arraycolsep}{5pt}
\vspace{-5pt}

\noindent
it follows from Lemma \ref{Qlemma} that the first term on the RHS of (\ref{expanded}) reduces to

\setlength{\arraycolsep}{0.3em}
\begin{equation} \label{reduced}
\begin{array}{rcl}
x_p^TR^T(\theta) QR(\theta)x_p & = & x_p^T R^T(\theta) R(\theta) Q x_p \vspace{5pt} \\
 & = & x_p^TQx_p
\end{array}
\end{equation}
\setlength{\arraycolsep}{5pt}

\noindent
so that from (\ref{auxCost1}), (\ref{expanded}) and (\ref{reduced}), $J_1^*(\theta)$ reduces to

\vspace{-5pt}
\setlength{\arraycolsep}{0.3em}
\begin{eqnarray} \label{simplified}
J_1^*(\theta) = \frac{1}{2} ( x_p^T Q x_p + x_q^TQx_q - 2 x_q^T R(\theta) Q x_p ).
\end{eqnarray}
\setlength{\arraycolsep}{5pt}
\vspace{-5pt}

\noindent
Notice in the last term on the RHS of (\ref{simplified}) that

\vspace{-5pt}
\setlength{\arraycolsep}{0.3em}
\begin{equation} \label{RcRs}
\begin{array}{rcl}
x_q^T R(\theta) Q x_p & = & x_q^T \left[ R_c(\theta) + R_s(\theta) \right] Q x_p \vspace{5pt} \\
& = & x_q^T R_c(\theta) Q x_p + x_q^T R_s(\theta)Q x_p \\
\end{array}
\end{equation}
\setlength{\arraycolsep}{5pt}

\noindent
where $R_c(\theta) = (I_n \otimes I_c) \cos(\theta)$ and $R_s(\theta) = (I_{n} \otimes I_s) \sin(\theta)$, meaning that

\vspace{-8pt}
\setlength{\arraycolsep}{0.3em}
\begin{eqnarray*}
-2 x_q^T R(\theta) Q x_p =  2\alpha \cos(\theta) + 2\beta \sin(\theta)
\end{eqnarray*}
\setlength{\arraycolsep}{5pt}
\vspace{-8pt}

\noindent
where $\alpha = - x_q^T (I_n \otimes I_c) Q x_p$ and $\beta = - x_q^T (I_n \otimes I_s) Q x_p$, from which it immediately follows that

\vspace{-5pt}
\setlength{\arraycolsep}{0.3em}
\begin{equation} \label{sinusoidalForm}
\begin{array}{rcl}
J_1^*(\theta) = \alpha \cos(\theta) + \beta \sin(\theta) + \gamma
\end{array}
\end{equation}
\setlength{\arraycolsep}{5pt}
\vspace{-8pt}

\noindent
where $\gamma = \frac{1}{2} \left (x_p^T Q x_p + x_q^TQx_q \right)$.  By virtue of the sinusoidal form (\ref{sinusoidalForm}), it follows that $J_1^*(\theta)$ has a unique minimum for $\theta \in [-\pi,\pi]$ given by (\ref{isoTheta}), which leads to the MLEs  of $t$ and $\mu$ given by (\ref{isoT}) and (\ref{isoFusion}), respectively.
\end{proof}
\vspace{10pt}


{
\bibliographystyle{ieeetr}
\bibliography{bibDatabase}
}

\end{document}